\documentclass[runningheads]{llncs}
\usepackage{amsmath,amssymb}
\usepackage[nobreak]{cite}
\usepackage[dvipdfmx]{graphicx}

\begin{document}
\title{Competitive Analysis for Two Variants of Online Metric Matching Problem\thanks{This work was partially supported by the joint project of Kyoto University and Toyota Motor Corporation, titled ``Advanced Mathematical Science for Mobility Society'' and JSPS KAKENHI Grant Numbers JP16K00017 and JP20K11677.} }
%
%\titlerunning{Competitive Analysis for Two Variants of Online Metric Matching Problem}
\author{Toshiya Itoh\inst{1} \and Shuichi Miyazaki\inst{2} \and Makoto Satake\inst{3}}
\authorrunning{T. Itoh et al.}
\institute{Department of Mathematical and Computing Science, Tokyo Institute of Technology, 2-12-1 Ookayama, Meguro-ku, Tokyo 152-8550, Japan \\ \email{titoh@c.titech.ac.jp}
\and
Academic Center for Computing and Media Studies, Kyoto University, Yoshida-Honmachi, Sakyo-ku, Kyoto 606-8501, Japan \\
\email{shuichi@media.kyoto-u.ac.jp} \and
Graduate School of Informatics, Kyoto University, \\ Yoshida-Honmachi, Sakyo-ku, Kyoto 606-8501, Japan \\
\email{satake@net.ist.i.kyoto-u.ac.jp}}

\sloppy
\maketitle              % typeset the header of the contribution
\begin{abstract}
In this paper, we study two variants of the online metric matching problem.
The first problem is the online metric matching problem where all the servers are placed at one of two positions in the metric space.
We show that a simple greedy algorithm achieves the competitive ratio of 3 and give a matching lower bound. 
The second problem is the online facility assignment problem on a line, where servers have capacities, servers and requests are placed on 1-dimensional line, and the distances between any two consecutive servers are the same.
We show lower bounds $1+ \sqrt{6}$ $(> 3.44948)$, $\frac{4+\sqrt{73}}{3}$ $(>4.18133)$ and $\frac{13}{3}$ $(>4.33333)$ on the competitive ratio when the numbers of servers are 3, 4 and 5, respectively.

\end{abstract}
\section{Introduction}
The online metric matching problem was introduced independently by Kalyanasundaram and Pruhs \cite{kalyanasundaram93} and Khuller, Mitchell and Vazirani \cite{khuller}.
In this problem, $n$ {\em servers} are placed on a given metric space.
Then $n$ {\em requests}, which are points on the metric space, are given to the algorithm one-by-one in an online fashion.
The task of an online algorithm is to match each request immediately to one of $n$ servers.
If a request is matched to a server, then it incurs a cost which is equivalent to the distance between them.
The goal of the problem is to minimize the sum of the costs.
The papers \cite{kalyanasundaram93} and \cite{khuller} presented a deterministic online algorithm (called {\em Permutation} in \cite{kalyanasundaram93}) and showed that it is $(2n-1)$-competitive and optimal.

In 1998, Kalyanasundaram and Pruhs \cite{kalyanasundaram98} posed a question whether we can have a better competitive ratio by restricting the metric space to a line, and introduced the problem called the {\em online matching problem on a line}.
Since then, this problem has been extensively studied, but there still remains a large gap between the best known lower bound 9.001 \cite{fuchs} and upper bound $O(\log n)$ \cite{raghvendra18} on the competitive ratio.

In 2020, Ahmed, Rahman and Kobourov \cite{ahmed} proposed a problem called the {\em online facility assignment problem} and considered it on a line, which we denote {\em OFAL} for short.
In this problem, all the servers (which they call {\em facilities}) and requests (which they call {\em customers}) lie on a 1-dimensional line, and the distance between every pair of adjacent servers is the same.
Also, each server has a {\em capacity}, which is the number of requests that can be matched to the server.
In their model, all the servers are assumed to have the same capacity.
Let us denote OFAL($k$) the OFAL problem where the number of servers is $k$. 
Ahmed et al. \cite{ahmed} showed that for OFAL($k$) the greedy algorithm is $4k$-competitive for any $k$ and a deterministic algorithm {\em Optimal-fill} is $k$-competitive for any $k>2$.

\subsection{Our contributions}

In this paper, we study a variant of the online metric matching problem where all the servers are placed at one of two positions in the metric space.
This is equivalent to the case where there are two servers with capacities.
We show that a simple greedy algorithm achieves the competitive ratio of 3 for this problem, and show that any deterministic online algorithm has competitive ratio at least 3.
 
We also study OFAL($k$) for small $k$.
Specifically, we show lower bounds $1+ \sqrt{6}$ $(> 3.44948)$, $\frac{4+\sqrt{73}}{3}$ $(>4.18133)$ and $\frac{13}{3}$ $(>4.33333)$ on the competitive ratio for OFAL($3$), OFAL($4$) and OFAL($5$), respectively.
We remark that our lower bounds $1+ \sqrt{6}$ for OFAL($3$) and $\frac{4+\sqrt{73}}{3}$ for OFAL($4$) do not contradict the above-mentioned upper bound of Optimal-fill, since upper bounds by Ahmed et al.~\cite{ahmed} are with respect to the {\em asymptotic} competitive ratio, while our lower bounds are with respect to the {\em strict} competitive ratio (see Sec. \ref{subsec:CR}).
%Hence, for example, our lower bound 3 for OFAL($3$) does not necessarily imply optimality of Optimal-fill.

\subsection{Related work}

In 1990, Karp, Vazirani and Vazirani \cite{karp} first studied an online version of the matching problem.
They studied the online matching problem on unweighted bipartite graphs with $2n$ vertices that contain a perfect matching, where the goal is to maximize the size of the obtained matching.
In \cite{karp}, they first showed that a deterministic greedy algorithm is $\frac{1}{2}$-competitive and optimal. 
They also presented a randomized algorithm {\em Ranking} and showed that it is $(1-\frac{1}{e})$-competitive and optimal. See \cite{mehta} for a survey of the online matching problem.

As mentioned before, Kalyanasundaram and Pruhs \cite{kalyanasundaram93} studied the online metric matching problem and showed that the algorithm {\em Permutation} is $(2n-1)$-competitive and optimal. Probabilistic algorithms for this problem were studied in \cite{meyerson, bansal}.

Kalyanasundaram and Pruhs \cite{kalyanasundaram98} studied the online matching problem on a line. They gave two conjectures that the competitive ratio of this problem is 9 and that the {\em Work-Function} algorithm has a constant competitive ratio, both of which were later disproved in \cite{koutsoupias} and \cite{fuchs}, respectively. This problem was studied in \cite{antoniadis14, raghvendra16, nayyar, raghvendra18, antoniadis18, gupta}, and the best known deterministic algorithm is the {\em Robust Matching} algorithm \cite{raghvendra16}, which is $\Theta(\log n)$-competitive \cite{nayyar, raghvendra18}. 
%However, an optimal algorithm is not yet known.

Besides the problem on a line, Ahmed, Rahman and Kobourov \cite{ahmed} studied the online facility assignment problem on an unweighted graph $G(V, E)$.  They showed that the greedy algorithm is $2|E|$-competitive and {\em Optimal-Fill} is $\frac{|E|k}{r}$-competitive, where $|E|$ is the number of edge of $G$ and $r$ is the radius of $G$.

\section{Preliminaries}

In this section, we give definitions and notations.

\subsection{Online metric matching problem with two servers}\label{subsec:1}

We define the online metric matching problem with two servers, denoted OMM($2$) for short.
Let $(X, d)$ be a metric space, where $X$ is a (possibly infinite) set of points and $d( \cdot, \cdot)$ is a distance function.
Let $S = \{s_1, s_2 \}$ be a set of servers and $R=\{r_1, r_2, \ldots, r_n\}$ be a set of requests.
A server $s_i$ is characterized by the position $p(s_i) \in X$ and the capacity $c_i$ that satisfies $c_1 + c_2 = n$.
This means that $s_i$ can be matched with at most $c_i$ requests ($i=1, 2$).
A request $r_i$ is also characterized by the position $p(r_i) \in X$.

$S$ is given to an online algorithm in advance, while requests are given one-by-one from $r_1$ to $r_n$.
At any time of the execution of an algorithm, a server is called {\em free} if the number of requests matched with it is less than its capacity, and {\em full} otherwise.
When a request $r_i$ is revealed, an online algorithm must match $r_i$ with one of free servers.
If $r_{i}$ is matched with the server $s_{j}$, the pair $(r_i, s_j)$ is added to the current matching and the cost $d(r_i, s_j)$ is incurred for this pair.
The cost of the matching is the sum of the costs of all the pairs contained in it.
The goal of OMM($2$) is to minimize the cost of the final matching.

\subsection{Online facility assignment problem on a line}

We give the definition of the online facility assignment problem on a line with $k$ servers, denoted OFAL($k$).
We state only differences from Sec.~\ref{subsec:1}.
The set of servers is $S = \{s_1, s_2, \ldots, s_k \}$ and all the servers have the same capacity $\ell$, i.e., $c_{i}=\ell$ for all $i$.
The number of requests must satisfy $n \leq \sum_{i=1}^{k} c_i = k \ell$.
All the servers and requests are placed on a real number line, so their positions are expressed by a real, i.e., $p(s_i) \in \mathbb{R}$ and $p(r_j) \in \mathbb{R}$.
Accordingly, the distance function is written as $d(r_i, s_j)=| p(r_i) - p(s_j) |$.
We assume that the servers are placed in an increasing order of their indices, i.e., $p(s_{1}) \leq p(s_{2}) \leq \ldots \leq p(s_{k})$.
In this problem, any distance between two consecutive servers is the same, that is, $| p(s_i) - p(s_{i+1}) | = d$ ($1 \le i \le n-1$) for some constant $d$.
Without loss of generality, we let $d=1$.

%Define $\sigma = ((X, d), S, R)$, where $\sigma$ is an input and $(X, d)$ is the metric space.
%In this problem, $X = \mathbb{R}$ and $d(r, s) = | r - s |$.

%Note that we are considering a many-to-one matching model since servers have capacities. However, this is equivalent to a one-to-one matching model by regarding a server with capacity $c$ as $c$ servers at the same location.

\subsection{Competitive ratio}\label{subsec:CR}

To evaluate the performance of an online algorithm, we use the {\em strict competitive ratio}.
(Hereafter, we omit ``strict''.)
For an input $\sigma$, let $ALG(\sigma)$ and $OPT(\sigma)$ be the costs of the matchings obtained by an online algorithm $ALG$ and an optimal offline algorithm $OPT$, respectively.
Then the competitive ratio of $ALG$ is the supremum of $c$ that satisfies $\frac{ALG(\sigma)}{OPT(\sigma)} \le c$ for any input $\sigma$.

\section{Online Metric Matching Problem with Two Servers}

\subsection{Upper bound}

In this section, we define a greedy algorithm $GREEDY$ for OMM($2$) and show that it is 3-competitive.

\begin{definition}\label{greedy}
When a request is given, $GREEDY$ matches it with the closest free server. If a given request is equidistant from the two servers and both servers are free, $GREEDY$ matches this request with $s_1$. 
\end{definition}

In the following discussion, we fix an optimal offline algorithm $OPT$.
If a request $r$ is matched with the server $s_x$ by $GREEDY$ and with $s_y$ by $OPT$, we say that $r$ is of {\em type} $\langle s_x, s_y \rangle$.
We then define some properties of inputs.

\begin{definition}\label{antiOpt}
Let $\sigma$ be an input to OMM($2$).  If every request in $\sigma$ is matched with a different server by $GREEDY$ and $OPT$, $\sigma$ is called {\em anti-opt}.
\end{definition}

\begin{definition}\label{sidePriority}
Let $\sigma$ be an input to OMM($2$).
Suppose that $GREEDY$ matches its first request $r_{1}$ to the server $s_{x} \in \{ s_{1}, s_{2} \}$. 
If $GREEDY$ matches $r_{1}$ through $r_{c_{x}}$ to $s_{x}$ (note that $c_{x}$ is the capacity of $s_{x}$) and $r_{c_{x}+1}$ through $r_{n}$ to the other server $s_{3-x}$, $\sigma$ is called {\em one-sided-priority}.
\end{definition}

For an input $\sigma$, we define $Rate(\sigma) = \frac{GREEDY(\sigma)}{OPT(\sigma)}$.
By the following two lemmas, we show that it suffices to consider inputs that are anti-opt and one-sided-priority.
We then show that $GREEDY$ is 3-competitive for such inputs.

\begin{lemma}\label{onlyAntiOpt}
For any input $\sigma$, there exists an anti-opt input $\sigma'$ such that $Rate(\sigma') \ge Rate(\sigma)$.
\end{lemma}

\begin{proof}
If $\sigma$ is already anti-opt, we can set $\sigma' = \sigma$. Hence, in the following, we assume that $\sigma$ is not anti-opt.
Then there exists a request $r$ in $\sigma$ that is matched with the same server $s_x$ by $OPT$ and $GREEDY$. Let $\sigma''$ be an input obtained from $\sigma$ by removing $r$ and subtracting the capacity of $s_x$ by 1. By this modification, neither $OPT$ nor $GREEDY$ changes a matching for the remaining requests.
Therefore, 
\begin{eqnarray*}
Rate(\sigma'') & = & \frac{GREEDY(\sigma) -d(r, s_{x})}{OPT(\sigma) - d(r, s_{x})} \\
 &\ge& \frac{GREEDY(\sigma)}{OPT(\sigma)}  \\
& = & Rate(\sigma).
\end{eqnarray*}

Let $\sigma'$ be the input obtained by repeating this operation until the input sequence becomes anti-opt. Then $\sigma'$ satisfies the conditions of this lemma.
\qed
\end{proof}

\begin{lemma}\label{onlySidePriority}
For any anti-opt input $\sigma$, there exists an anti-opt and one-sided-priority input $\sigma'$ such that $Rate(\sigma') \ge Rate(\sigma)$.
\end{lemma}

\begin{proof}

If $\sigma$ is already one-sided-priority, we can set $\sigma' = \sigma$. Hence, in the following, we assume that $\sigma$ is not one-sided-priority.

Since $\sigma$ is anti-opt, $\sigma$ contains only requests of type $\langle s_1, s_2 \rangle$ or $\langle s_2, s_1 \rangle$. Without loss of generality, assume that in execution of $GREEDY$, the server $s_{1}$ becomes full before $s_{2}$, and let $r_{t}$ be the request that makes $s_1$ full (i.e., $r_{t}$ is the last request of type $\langle s_1, s_2 \rangle$).

Because $\sigma$ is not one-sided-priority, $\sigma$ includes at least one request $r_{i}$ of type $\langle s_2, s_1 \rangle$ before $r_{t}$. Let $\sigma''$ be the input obtained from $\sigma$ by moving $r_{i}$ to just after $r_{t}$. 
Since the set of requests is unchanged in $\sigma$ and $\sigma''$, an optimal matching for $\sigma$ is also optimal for $\sigma''$, so $OPT(\sigma'')=OPT(\sigma)$.
In the following, we show that $GREEDY$ matches each request to the same server in $\sigma$ and $\sigma''$. 
The sequence of requests up to $r_{i-1}$ are the same in $\sigma''$ and $\sigma$, so the claim clearly holds for $r_{1}$ through $r_{i-1}$.
The behavior of $GREEDY$ for $r_{i+1}$ through $r_{t}$ in $\sigma''$ is also the same for those in $\sigma$, because when serving these requests, both $s_{1}$ and $s_{2}$ are free in both $\sigma$ and $\sigma''$.
Just after serving $r_{t}$ in $\sigma''$, $s_1$ becomes full, so $GREEDY$ matches $r_{i}, r_{t+1}, \ldots, r_{n}$ with $s_{2}$ in $\sigma''$.
Note that these requests are also matched with $s_{2}$ in $\sigma$.
Hence $GREEDY(\sigma'')=GREEDY(\sigma)$ and it results that $Rate(\sigma'') = Rate(\sigma)$.
Note that $\sigma''$ remains anti-opt.

Let $\sigma'$ be the input obtained by repeating this operation until the input sequence becomes one-sided-priority. Then $\sigma'$ satisfies the condition of the lemma.
\qed
\end{proof}

We can now prove the upper bound.

\begin{theorem}\label{upperBound}
The competitive ratio of $GREEDY$ is at most 3 for OMM($2$).
\end{theorem}

\begin{proof}
By Lemma \ref{onlyAntiOpt}, it suffices to analyze only anti-opt inputs. 
In an anti-opt input, the number of requests of type $\langle s_1, s_2 \rangle$ and that of type $\langle s_2, s_1 \rangle$ are the same and the capacities of $s_{1}$ and $s_{2}$ are $n/2$ each.
By Lemma \ref{onlySidePriority}, it suffices to analyze only the inputs where the first $n/2$ requests are of type $\langle s_1, s_2 \rangle$ and the remaining $n/2$ requests are of type $\langle s_2, s_1 \rangle$.

Let $\sigma$ be an arbitrary such input.
Then we have that
\[GREEDY(\sigma) = \sum_{i=1}^{n/2} d(r_{i}, s_{1}) + \sum_{i=n/2+1}^n d(r_{i}, s_{2}) \] and
\[OPT(\sigma) = \sum_{i=1}^{n/2} d(r_{i}, s_{2}) + \sum_{i=n/2+1}^n d(r_{i}, s_{1}). \]

When serving $r_{1}, r_{2}, \ldots, r_{n/2}$, both servers are free but GREEDY matched them with $s_{1}$.
Hence $d(r_{i}, s_{1}) \leq d(r_{i}, s_{2})$ for $1 \leq i \leq n/2$.
By the triangle inequality, we have $d(r_{i}, s_{2}) \leq  d(s_{1}, s_{2}) + d(r_{i}, s_{1})$ for  $n/2+1 \leq i \leq n$.
Again, by the triangle inequality, we have $d(s_{1}, s_{2}) \leq d(r_{i}, s_{1}) + d(r_{i}, s_{2})$ for $1 \leq i \leq n$.

From these inequalities, we have that
\begin{eqnarray*}
GREEDY(\sigma) & = & \sum_{i=1}^{n/2} d(r_{i}, s_{1}) + \sum_{i=n/2+1}^n d(r_{i}, s_{2}) \\
  & \leq & \sum_{i=1}^{n/2} d(r_{i}, s_{2}) + \sum_{i=n/2+1}^n (d(s_{1}, s_{2}) + d(r_{i}, s_{1}))  \\
   & = & OPT(\sigma) + \frac{n}{2} d(s_{1}, s_{2})  \\
   & = & OPT(\sigma) + \frac{1}{2}\sum_{i=1}^{n} d(s_{1}, s_{2}) \\
   & \leq & OPT(\sigma) + \frac{1}{2}\sum_{i=1}^{n} (d(r_{i}, s_{1}) + d(r_{i}, s_{2}))  \\
   & = & OPT(\sigma) + \frac{1}{2}(OPT(\sigma)+GREEDY(\sigma))  \\
   & = & \frac{3}{2}OPT(\sigma)+\frac{1}{2}GREEDY(\sigma).  \\
\end{eqnarray*}
Thus $GREEDY(\sigma) \leq 3OPT(\sigma)$ and the competitive ratio of $GREEDY$ is at most 3.
\qed
\end{proof}

\subsection{Lower bound}

\begin{theorem}\label{lowerBound}
The competitive ratio of any deterministic online algorithm for OMM($2$) is at least 3.
\end{theorem}

\begin{proof}
We prove this lower bound on a 1-dimensional real line metric.
Let $p(s_{1})=-d$ and $p(s_{2})=d$ for a constant $d$.
Consider any deterministic algorithm $ALG$.
First, our adversary gives $c_1-1$ requests at $p(s_1)$ and $c_2-1$ requests at $p(s_2)$.
$OPT$ matches the first $c_1-1$ requests with $s_1$ and the rest with $s_2$.
If there exists a request that $ALG$ matches differently from $OPT$, the adversary gives two more requests, one at $p(s_1)$ and the other at $p(s_2)$.
Then, the cost of $OPT$ is zero, while the cost of $ALG$ is positive, so the ratio of them becomes infinity.

Next, suppose that $ALG$ matches all these requests with the same server as $OPT$.
Then the adversary gives the next request at the origin $0$.
Let $s_{x}$ be the server that $ALG$ matches this request with.
Then $OPT$ matches this request with the other server $s_{3-x}$.
After that, the adversary gives the last request at $p(s_{x})$.
$ALG$ has to match it with $s_{3-x}$ and $OPT$ matches it with $s_{x}$.
The costs of $ALG$ and $OPT$ for this input is $3d$ and $d$, respectively.
This completes the proof.
\qed
\end{proof}

\section{Online Facility Assignment Problem on Line}

In this section, we show lower bounds on the competitive ratio of OFAL($k$) for $k=3, 4$, and 5.
To simplify the proofs, we recall useful properties that allow us to restrict online algorithms to consider \cite{koutsoupias, antoniadis18}.
When a request $r$ is given, the {\em surrounding servers} for $r$ are the closest free server to the left of $r$ and the closest free server to the right of $r$.
If, for any input, an algorithm $ALG$ matches every request with one of the surrounding servers, $ALG$ is called {\em surrounding-oriented}.

\begin{proposition}\label{surrounding}
%\begin{lemma}\label{surrounding}
For any algorithm $ALG$, there exists a surrounding-oriented algorithm $ALG'$ such that $ALG'(\sigma) \le ALG(\sigma)$ for any input $\sigma$.
%\end{lemma}
\end{proposition}

%\begin{figure}
%\includegraphics[width=\textwidth]{img/OFAS3SS.png}
%\caption{The behavior of our adversary in the proof of Theorem \ref{lb3}.} \label{OFAS3}
%\end{figure}

%\begin{proof}
%If $ALG$ is surrounding-oriented, we can set $ALG'$ to $ALG$. Hence, in the following, we assume that $ALG$ is not surrounding-oriented.
%
%Since $ALG$ is not surrounding-oriented, there exits an input $\sigma$ such that $ALG$ matches at least one request with a non-surrounding server. Let $r_k$ be the earliest request of such requests $(1 \le k \le n-1)$.
%
%We construct algorithm $\overline{ALG}$ based on $ALG$. Let $s_k$ be the server that $ALG$ matches $r_k$ with and $s_k'$ be the closest free server on the same side as $s_k$. Let $r_k'$ be the request that $ALG$ matches with $s_k'$. At this time, $\overline{ALG}$ matches $r_k$ with $s_k'$ and $r_k'$ with $s_k$. Except for these, $\overline{ALG}$ performs same as $ALG$.
%
%Next we consider these cost for $\sigma$. Without loss of generality, assume that $r_k < s_k$. If $s_k' < r_k'$, then $\overline{ALG}$ breaks the cross of the pair of $r_k$ and $s_k$ and the pair of $r_k'$ and $s_k'$. Therefore, the cost by $\overline{ALG}$ is less than the cost by $ALG$. If $s_k' \le r_k'$, then requests are on the left side of $s_k'$. Therefore, the cost by $\overline{ALG}$ is equal to the cost by $ALG$. From the above, the cost by $\overline{ALG}$ is less than or equal to that by $ALG$.
%
%Let $ALG'$ be the algorithm obtained by repeating this operation until time $n-1$. $ALG'$ is what we need. 
%\qed
%\end{proof}

By Proposition \ref{surrounding}, it suffices to consider only surrounding-oriented algorithms for lower bound arguments.

\begin{theorem}\label{lb3}
The competitive ratio of any deterministic online algorithm for OFAL($3$) is at least $1+ \sqrt{6}$ $(> 3.44948)$.
\end{theorem}

\begin{proof}
Let $ALG$ be any surrounding-oriented algorithm.
Our adversary first gives $\ell - 1$ requests at $p(s_i)$ for each $i=1,2$ and 3.
$OPT$ matches every request $r$ with the server at the same position $p(r)$. 
If $ALG$ matches some request $r$ with a server not at $p(r)$, then the adversary gives three more requests, one at each position of the server. The cost of $ALG$ is positive and the cost of $OPT$ is zero, so the ratio of the costs is infinity.

Next, suppose that $ALG$ matches all these requests to the same server as $OPT$.
Let $x=\sqrt{6}-2$ $(\simeq 0.44949)$ and $y=3\sqrt{6}-7$ $(\simeq 0.34847)$.
The adversary gives a request $r_1$ at $p(s_2)+x$.

\smallskip
\smallskip

\noindent
{\bf\boldmath Case 1. $ALG$ matches $r_1$ with $s_3$.}\\
See Fig.~\ref{fig:3-1}.  The adversary gives the next request $r_2$ at $p(s_3)$.
$ALG$ matches it with $s_2$.
Finally, the adversary gives a request $r_{3}$ at $p(s_1)$ and $ALG$ matches it with $s_1$.
The cost of $ALG$ is $2-x=4-\sqrt{6}$ and the cost of $OPT$ is $x=\sqrt{6}-2$.
The ratio is $\frac{4-\sqrt{6}}{\sqrt{6}-2} = 1+ \sqrt{6}$.

\begin{figure}
\begin{center}
\includegraphics[width=7cm]{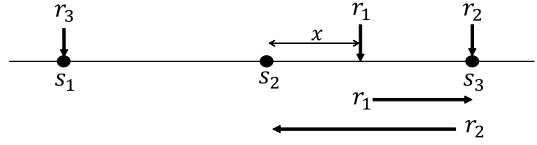}
\end{center}
\caption{Requests and $ALG$'s matching for Case 1 of Theorem \ref{lb3}.} \label{fig:3-1}
\end{figure}

\smallskip
\smallskip

\noindent
{\bf\boldmath Case 2. $ALG$ matches $r_1$ with $s_2$.}\\
The adversary gives the next request $r_2$ at $p(s_2)-y$. 
We have two subcases.

\smallskip
\smallskip

\noindent
{\bf\boldmath  Case 2-1. $ALG$ matches $r_2$ with $s_1$.}\\
See Fig.~\ref{fig:3-2-1}. The adversary gives a request $r_3$ at $p(s_1)$ and $ALG$ matches it with $s_{3}$.
The cost of $ALG$ is $3+x-y=8-2\sqrt{6}$ and the cost of $OPT$ is $1-x+y=2\sqrt{6}-4$.
The ratio is $\frac{8-2\sqrt{6}}{2\sqrt{6}-4} = 1+ \sqrt{6}$.

\begin{figure}
\begin{center}
\includegraphics[width=7cm]{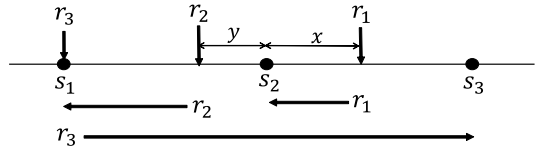}
\end{center}
\caption{Requests and $ALG$'s matching for Case 2-1 of Theorem \ref{lb3}.} \label{fig:3-2-1}
\end{figure}

\smallskip
\smallskip

\noindent
{\bf\boldmath  Case 2-2. $ALG$ matches $r_2$ with $s_3$.}\\
See Fig.~\ref{fig:3-2-2}. The adversary gives a request $r_3$ at $p(s_3)$ and $ALG$ matches it with $s_{1}$.
The cost of $ALG$ is $3+x+y=4\sqrt{6}-6$ and the cost of $OPT$ is $1+x-y=6-2\sqrt{6}$.
The ratio is $\frac{4\sqrt{6}-6}{6-2\sqrt{6}} = 1+ \sqrt{6}$.

\begin{figure}
\begin{center}
\includegraphics[width=7cm]{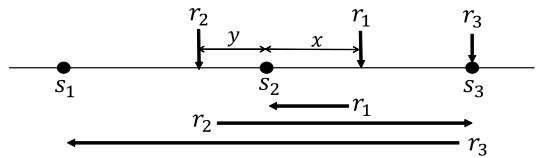}
\end{center}
\caption{Requests and $ALG$'s matching for Case 2-2 of Theorem \ref{lb3}.} \label{fig:3-2-2}
\end{figure}

In any case, the ratio of $ALG$'s cost to $OPT$'s cost is $1+ \sqrt{6}$.
This completes the proof.
\qed
\end{proof}

\begin{theorem}\label{lb4}
The competitive ratio of any deterministic online algorithm for OFAL($4$) is at least $\frac{4+\sqrt{73}}{3}$ $(>4.18133)$.
\end{theorem}

\begin{proof}
Let $ALG$ be any surrounding-oriented algorithm.
In the same way as the proof of Theorem \ref{lb3}, the adversary first gives $\ell - 1$ requests at $p(s_i)$ for $i=1, 2, 3$, and 4, and we can assume that $OPT$ and $ALG$ match each of these requests to the server at the same position.
Then, the adversary gives a request $r_1$ at $\frac{p(s_2) + p(s_3)}{2}$.
Without loss of generality, assume that $ALG$ matches it with $s_2$.

Let $x=\frac{10-\sqrt{73}}{2}$ $(\simeq 0.72800)$ and $y=\frac{11\sqrt{73}-93}{8}$ $(\simeq 0.12301)$.
The adversary gives a request $r_2$ at $p(s_1)+x$.
We consider two cases depending on the behavior of $ALG$.

\smallskip
\smallskip

\noindent
{\bf\boldmath Case 1. $ALG$ matches $r_2$ with $s_1$.}\\
See Fig.~\ref{fig:4-1}.  The adversary gives the next request $r_3$ at $p(s_1)$.
$ALG$ has to match it with $s_3$.
Finally, the adversary gives a request $r_{4}$ at $p(s_4)$ and $ALG$ matches it with $s_4$. 
The cost of $ALG$ is $\frac{5}{2}+x=\frac{15-\sqrt{73}}{2}$ and the cost of $OPT$ is $\frac{3}{2}-x=\frac{\sqrt{73}-7}{2}$.
The ratio is $\frac{15-\sqrt{73}}{\sqrt{73}-7} = \frac{4+ \sqrt{73}}{3}$.

\begin{figure}
\begin{center}
\includegraphics[width=7cm]{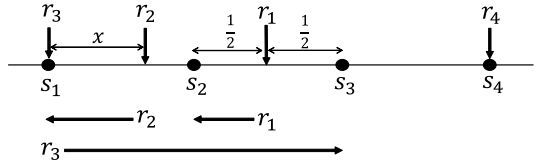}
\end{center}
\caption{Requests and $ALG$'s matching for Case 1 of Theorem \ref{lb4}.} \label{fig:4-1}
\end{figure}

\noindent
{\bf\boldmath  Case 2. $ALG$ matches $r_2$ with $s_3$.}\\
The adversary gives the next request $r_3$ at $p(s_3)+y$. 
We have two subcases.

\smallskip
\smallskip

\noindent
{\bf\boldmath  Case 2-1. $ALG$ matches $r_3$ with $s_4$.}\\
See Fig.~\ref{fig:4-2-1}. The adversary gives a request $r_4$ at $p(s_4)$. $ALG$ has to match it with $s_1$. The cost of $ALG$ is $\frac{13}{2}-x-y=\frac{105-7\sqrt{73}}{8}$ and the cost of $OPT$ is $\frac{1}{2}+x+y=\frac{7\sqrt{73}-49}{8}$.
The ratio is $\frac{105-7\sqrt{73}}{7\sqrt{73}-49} = \frac{4+ \sqrt{73}}{3}$.

\begin{figure}
\begin{center}
\includegraphics[width=7cm]{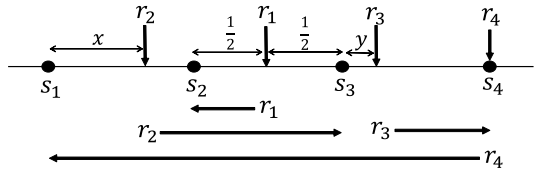}
\end{center}
\caption{Requests and $ALG$'s matching for Case 2-1 of Theorem \ref{lb4}.} \label{fig:4-2-1}
\end{figure}

\smallskip
\smallskip

\noindent
{\bf\boldmath  Case 2-2. $ALG$ matches $r_3$ with $s_1$.}\\
See Fig.~\ref{fig:4-2-2}. The adversary gives a request $r_4$ at $p(s_1)$ and $ALG$ has to match it with $s_4$.
The cost of $ALG$ is $\frac{15}{2}-x+y=\frac{15\sqrt{73}-73}{8}$ and the cost of $OPT$ is $\frac{5}{2}-x-y=\frac{73-7\sqrt{73}}{8}$.
The ratio is $\frac{15\sqrt{73}-73}{73-7\sqrt{73}} = \frac{4+ \sqrt{73}}{3}$.

\begin{figure}
\begin{center}
\includegraphics[width=7cm]{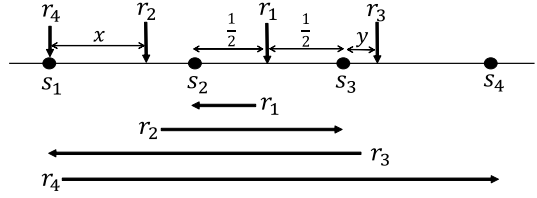}
\end{center}
\caption{Requests and $ALG$'s matching for Case 2-2 of Theorem \ref{lb4}.} \label{fig:4-2-2}
\end{figure}

\smallskip
\smallskip

In any case, the ratio of $ALG$'s cost to $OPT$'s cost is $\frac{4+ \sqrt{73}}{3}$.
This completes the proof.
\qed
\end{proof}

%\begin{figure}
%\includegraphics[width=\textwidth]{img/OFAS4SS.png}
%\caption{The behavior of adversary in the proof of Theorem \ref{lb4}.} \label{OFAS4}
%\end{figure}

\begin{theorem}\label{lb5}
The competitive ratio of any deterministic online algorithms for OFAL($5$) is at least $\frac{13}{3}$ $(>4.33333)$.
\end{theorem}

\begin{proof}
Let $ALG$ be any surrounding-oriented algorithm.
In the same way as the proof of Theorem \ref{lb3}, the adversary first gives $\ell - 1$ requests at $p(s_i)$ for $i=1, 2, 3, 4$, and 5, and we can assume that $OPT$ and $ALG$ match each of these requests to the server at the same position. 

Then, the adversary gives a request $r_1$ at $p(s_3)$. 
If $ALG$ matches this with $s_2$ or $s_4$, the adversary gives the remaining requests at $p(s_1)$, $p(s_2)$, $p(s_4)$ and $p(s_5)$.
$OPT$'s cost is zero, while $ALG$'s cost is positive, so the ratio is again infinity. 
Therefore, assume that $ALG$ matches $r_1$ with $s_3$. 
The adversary then gives a request $r_2$ at $p(s_3)$. 
Without loss of generality, assume that $ALG$ matches it with $s_2$. 
Next, the adversary gives a request $r_3$ at $p(s_1)+\frac{7}{8}$. 
We consider two cases depending on the behavior of $ALG$.

\smallskip
\smallskip

\noindent
{\bf\boldmath Case 1. $ALG$ matches $r_3$ with $s_1$.}\\
See Fig.~\ref{fig:5-1}.  The adversary gives the next request $r_4$ at $p(s_1)$. $ALG$ has to match it with $s_4$. Finally, the adversary gives a request $r_{5}$ at $p(s_5)$ and $ALG$ matches it with $s_5$. The cost of $ALG$ is $\frac{39}{8}$ and the cost of $OPT$ is $\frac{9}{8}$.
The ratio is $\frac{13}{3}$.

\smallskip
\smallskip

\begin{figure}
\begin{center}
\includegraphics[width=8cm]{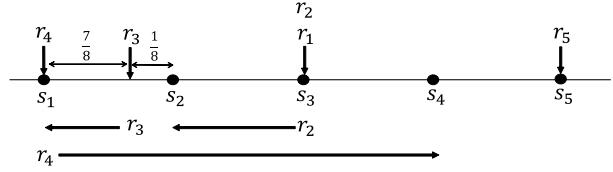}
\end{center}
\caption{Requests and $ALG$'s matching for Case 1 of Theorem \ref{lb5}.} \label{fig:5-1}
\end{figure}

\noindent
{\bf\boldmath Case 2. $ALG$ matches $r_3$ with $s_4$.}\\
The adversary gives the next request $r_4$ at $p(s_4)$.  We have two subcases.

\noindent
{\bf\boldmath Case 2-1. $ALG$ matches $r_4$ with $s_1$.}\\
See Fig.~\ref{fig:5-2-1}. The adversary gives a request $r_5$ at $p(s_1)$ and $ALG$ has to match it with $s_5$. The cost of $ALG$ is $\frac{81}{8}$ and the cost of $OPT$ is $\frac{17}{8}$.
The ratio is $\frac{81}{17}>\frac{13}{3}$.
\smallskip
\smallskip

\begin{figure}
\begin{center}
\includegraphics[width=8cm]{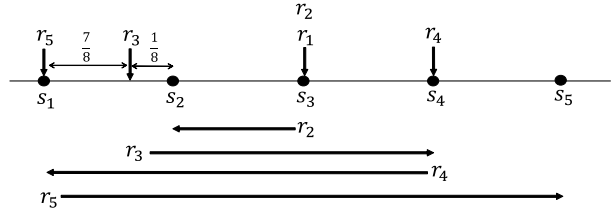}
\end{center}
\caption{Requests and $ALG$'s matching for Case 2-1 of Theorem \ref{lb5}.} \label{fig:5-2-1}
\end{figure}

\noindent
{\bf\boldmath Case 2-2. $ALG$ matches $r_4$ with $s_5$.}\\
See Fig.~\ref{fig:5-2-2}.  The adversary gives a request $r_5$ at $p(s_5)$ and $ALG$ has to match it with $s_1$. The cost of $ALG$ is $\frac{65}{8}$ and the cost of $OPT$ is $\frac{15}{8}$.
The ratio is $\frac{13}{3}$.

\begin{figure}
\begin{center}
\includegraphics[width=8cm]{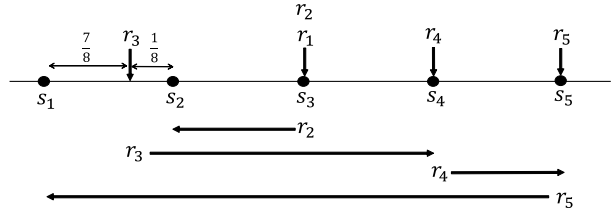}
\end{center}
\caption{Requests and $ALG$'s matching for Case 2-2 of Theorem \ref{lb5}.} \label{fig:5-2-2}
\end{figure}

\smallskip
\smallskip

In any case, the ratio of $ALG$'s cost to $OPT$'s cost is at least $\frac{13}{3}$, which completes the proof.
\qed
\end{proof}

\section{Conclusion}
%In this paper, we studied the online facility assignment problem on a line with small number of servers.
%For OFAL($2$), we showed that the greedy algorithm is 3-competitive and optimal.
%For OFAL($3$), OFAL($4$) and OFAL($5$), we showed lower bounds 3, $\frac{11}{3}$ and 4 respectively.
%
%One of the most interesting problems is to generalize the number of servers in lower bound arguments.

In this paper, we studied two variants of the online metric matching problem.
The first is a restriction where all the servers are placed at one of two positions in the metric space.
For this problem, we presented a greedy algorithm and showed that it is 3-competitive.
We also proved that any deterministic online algorithm has competitive ratio at least 3, giving a matching lower bound.
The second variant is the Online Facility Assignment Problem on a line with a small number of servers.
We showed lower bounds on the competitive ratio $1+ \sqrt{6}$, $\frac{4+\sqrt{73}}{3}$, and $\frac{13}{3}$ when the numbers of servers are 3, 4, and 5, respectively.

One of the future work is to analyze the online metric matching problem with three or more server positions.
Another interesting direction is to consider an optimal online algorithm for the Online Facility Assignment Problem on a line when the numbers of servers are 3, 4, and 5.

\bibliographystyle{splncs04}
\bibliography{ref}

\end{document}